\documentclass[12pt]{iopart}

\expandafter\let\csname equation*\endcsname\relax
\expandafter\let\csname endequation*\endcsname\relax

\usepackage{bigints}
\usepackage{amsfonts}
\usepackage{amsthm}
\usepackage{amsmath}
\usepackage{graphicx}
\usepackage{psfrag}
\usepackage{amssymb}
\usepackage{tikz}
\usetikzlibrary{arrows.meta}

\swapnumbers

\newtheorem{Thm}{Theorem}[section]

\newtheorem{Cor}[Thm]{Corollary}

\newtheorem*{Thm2.1.3}{{Theorem 2.1.3 of \cite{CoddingtonTheory}}}
\newtheorem*{ThmI.5.3}{{Theorem I.5.3 of \cite{Hale2009Ordinary}}}

\theoremstyle{definition}

\theoremstyle{remark}

\newcommand{\man}[1]{\ensuremath{\mathcal{#1}}} 
\newcommand{\bound}[2][]{\ensuremath{\partial{#1}\man{#2}}} 
\newcommand{\EmptySet}{\ensuremath{\varnothing}}
\newcommand{\R}{\ensuremath{\mathbb{R}}}
\newcommand{\N}{\ensuremath{\mathbb{N}}}

\begin{document}

  \title{The Endpoint Theorem{$^\text{\dag}$}}
  \author{Susan M Scott$^{1,2}$ and Ben E Whale$^{1}$}

  \address{$^1$ 
    Centre for Gravitational Astrophysics,
    Department of Quantum Science, 
    Research School of Physics,
    College of Science, 
    The Australian National University} 
  \address{$^2$
    Australian Research Council Centre of Excellence 
      for Gravitational Wave Discovery (OzGrav)}

  \ead{\mailto{susan.scott@anu.edu.au}
  \mailto{ben@benwhale.com}}

  \begin{abstract}
    The Endpoint Theorem links the existence of a sequence (curve), without
    accumulation points, in a manifold to the existence of an 
    open embedding of
    that manifold so that the image of the
    given sequence (curve) has a unique endpoint. 
    It plays a fundamental role in the theory
    of the Abstract Boundary as it implies that
    there is always an Abstract Boundary boundary point to represent the
    endpoint of such sequences and curves.
    The Endpoint Theorem will be of interest to researchers
    analysing specific spacetimes as it shows how to construct
    a chart in the original manifold which contains the sequence (curve).
    In particular, it has application to the 
    study of singularities predicted by the singularity theorems.
  \end{abstract}
  
  \submitto{\CQG}


  \footnotetext{This paper is dedicated to the memory of Christopher J.S. Clarke.}

  \maketitle

\section{Introduction}

  This paper addresses a technical question about the ability to build
  an
  open embedding, $\phi: \man{M}\to\man{M}_\phi$, 
  for a manifold $\man{M}$, in which the image of 
  a given accumulation point free sequence $(x_i)_{i\in\N}\subset\man{M}$,
  has a unique limit point in the larger manifold. This question has significant
  physical interest if one wishes, as we do, to understand the physical consequences
  of the Penrose-Hawking singularity theorems using only the context of those
  theorems. It is apparent that both the motivation for doing this and the ``rules of the game'' are no 
  longer well understood, which is disappointing, since a great deal of effort was spent
  in the 70s and 80s attempting to understand the implications of
  the singularity theorems in this way. As an example, 
  Christopher Clarke who worked on the original version of the Endpoint Theorem
  with Scott, produced a book \cite{Clarke1993}, and numerous other publications, for example
  \cite{SlupinskiClarke1980, TiplerClarkeEllis1980, Clarke1979, 1978CMaPh..58..291C,1975GReGr...6...35C,
  MR0371347, ClarkeSchmidt1977, Clarke1983, FamaClarke1998, ClarkeKrolak1985},
  which provide ample examples of what we mean by ``the context for the singularity
  theorems''. 
  Nevertheless, for the benefit of the reader, we provide here
  an explanation of these topics relevant for this paper. 
  The reader familiar with applications of
  differential topology to relativity (e.g.\ \cite{Penrose1972}), 
  or more broadly with
  global Lorentzian geometry (e.g.\ \cite{BeemEhrlichEasley1996}),
  may wish to skip directly to Section 
  \ref{sec_theorem}.

  We begin with a discussion of the 
  Penrose-Hawking singularity theorems. For more
  detail refer to \cite{senovilla1998singularity, senovilla20151965}
  for general overviews, \cite{Penrose1972, HawkingEllis1973} for self-contained
  presentations by Penrose, and Hawking and Ellis, respectively, or \cite[Chapter 12]{BeemEhrlichEasley1996} for
  a more modern, mathematically focussed, presentation.
  Theorem 12.43 of \cite{BeemEhrlichEasley1996} is an example of
  a modern singularity theorem that aligns with the spirit of the Penrose-Hawking 
  singularity theorems:
  \begin{Thm}[{\cite[Theorem 12.43]{BeemEhrlichEasley1996}}]
    \label{theorem_BEE}
    No Lorentzian manifold $(\man{M},g)$ of dimension greater than $2$ can satisfy
    all of the following three requirements together:
    \begin{enumerate}
      \item The manifold $(\man{M},g)$ has no closed timelike curves.
      \item There exists a future trapped or past trapped set in $(\man{M},g)$.
      \item Every inextendible, nonspacelike geodesic in $(\man{M},g)$ contains
        a pair of conjugate points.
    \end{enumerate}
  \end{Thm}
  The proof boils down to showing that the first and second conditions imply that
  there exists a geodesic without a pair of conjugate points. The third
  condition then gives a contradiction. 

  When one also assumes the weak energy condition, the Raychaudhuri 
  equation can be used to show that there exists an inextendible, incomplete,
  causal geodesic. If the manifold is assumed to be maximally extended,
  then one can claim that the inextendible, incomplete, causal geodesic
  ``approaches a singularity''. We write ``approaches a singularity'' in quotes
  because the predicted causal geodesic has no accumulation points.
  It ``approaches'' no points in the manifold. 
  Moreover, beyond defining a singularity to be an incomplete, inextendible
  curve in a maximally extended manifold, it is not clear exactly what
  ``singularity'' means. The singularity theorems do not predict any
  kind of physical experience for an observer on the incomplete, inextendible
  curve.
  Even today, more than 50 years after Penrose's paper \cite{penrose1965gravitational}
  (the basis for his receipt of the 2020 Nobel Prize in Physics), there
  is no accepted definition of what a singularity is \emph{in this context}.
  There are accepted definitions within classes of Lorentzian manifolds,
  for example certain parameter dependent solutions of Einstein's equations,
  but no accepted definition for the very general Lorentzian manifolds to which
  the singularity theorems apply.
  Interpretation and resolution of the Gordian knot of a phrase, ``approaches a singularity'', sums up the
  goal of the program of research in which this paper lies.

  Theorem \ref{theorem_BEE}
  is beautifully simple. It assumes very little, and only
  assumes things that have strong physical
  justifications \cite[Chapter 8]{HawkingEllis1973}. Yet, with the 
  additions
  of the weak energy condition and the assumption of maximal extension,
  it manages to show that the manifold is causal geodesically incomplete.
  This then has the interpretation that it is possible for a freely falling
  observer to exist for only a finite amount of time; a staggeringly profound
  statement about cosmology. 

  At the time when Penrose published the
  first singularity theorem \cite{penrose1965gravitational}, it was important that the
  theorem was very general. Two years earlier Lifshitz and Khalatnikov
  \cite{lifshitz1963investigations}
  had claimed that, ``An attempt is made to provide an answer to one of the
  principal questions of modern cosmology: ‘does the general solution of the
  gravitational equations have a singularity?’ The authors give a negative
  answer to this question.'' Because Penrose's result was very general,
  it demonstrated that this particular claim of Lifshitz and Khalatnikov
  could not be true. Six years later Belinskii, Lifshitz and Khalatnikov
  published their famous paper claiming to have derived the properties
  of generic cosmological singularities \cite{0038-5670-13-6-R04}.
  This paper
  birthed the BKL conjecture.
  It was the generality of Penrose's result that made it physically
  relevant. It is this generality that we seek to preserve in our study
  of gravitational singularities.

  Some researchers assume that it is necessary to involve Einstein's field
  equations to study what might happen to an observer along an incomplete,
  inextendible geodesic.
  They wish to phrase
  Einstein's equations as an initial value (or initial boundary value)
  problem and study the evolution of curvature via application of PDE 
  techniques.
  This assumes that the manifold is globally hyperbolic. Often
  a global splitting of the manifold, e.g.\ into a $3+1$ form, is also used.
  We see this as an interesting subcase of the full problem. 

  Within the context of PDE techniques, 
  ``approaches a singularity'' has a clear
  meaning and interpretation.
  For example,
  given a $3+1$ decomposition of the manifold, one has a canonical set
  of charts within which to perform computations. The properties and relations
  of these charts are known.
  These charts play a special role
  in supporting the analysis of the manifold. 
  In particular, points on the boundary of the images of the charts in Euclidean
  space are taken to play the role of the location of any singularity, 
  points at infinity, or points through which the manifold may be extended.
  We will call these boundary points.
  Thus the canonical charts define what ``approach'' means. 
  The curvature, or tensors related to the connection, can also be directly studied
  in the chart and ``on approach'' to boundary points. 
  The result is
  that one can determine properties of the singularity.

  In the more general context of the singularity theorems the manifold cannot be
  assumed to be globally hyperbolic. Any $3+1$ decomposition can only be assumed
  to be a local decomposition with respect to a particular chart. There is
  no reason why one can assume that some nice initial boundary value problem
  for Einstein's field equations can be given. There is no clear meaning
  nor interpretation for what ``approaches a singularity'' means.
  In particular, the lack of global charts means that one must take additional
  care in the use of chart based techniques. There is no {\em a priori} reason
  to select one chart's boundary points over another. 

  The inappropriateness of assumptions related to the application of PDE techniques
  in our work means that we must necessarily address questions that
  those assumptions usually render solved. Thus we must find other ways of
  solving these questions. Theorem \ref{Thm:Endpoint Theorem} is an example of
  this. In a globally hyperbolic manifold every sequence without accumulation points,
  in the manifold,
  has a set of boundary points to which it limits. This happens because of 
  global hyperbolicity. Theorem \ref{Thm:Endpoint Theorem} shows that the same
  is true for any manifold. For researchers unaware of the field of 
  differential topology in relativity, a result 
  like Theorem \ref{Thm:Endpoint Theorem}
  is likely to seem esoteric and unnecessary. We hope that the literature
  that we cite here can be used, by academics for which our result is
  esoteric, as a way to gain more understanding of other approaches to the
  open problems of general relativity.

  The field of boundary constructions in general relativity consists of
  attempts to build mathematical structures via which the
  phrase ``approaches a singularity'' can find valid, chart independent,
  interpretation. Even though a Lorentzian manifold has a metric, that metric
  does not induce a distance and thus the na\"{i}ve approach of using a Cauchy 
  completion fails; which distance on the manifold should you choose
  and why should you choose that distance? Sormani and Vega have 
  recently published a paper discussing this issue \cite{sormani2016null}.

  Modern examples of boundary constructions can be found in
  \cite{sormani2016null, flores2011final, garcia2003causal, ScottSzekeres1994, Whale2014Chart}.
  At a minimum, a boundary construction should produce a topological
  space into which the manifold $\man{M}$ can be embedded. The topological space
  should involve as few choices as possible and should be geometrically defined.
  Points in the topological space that are not in the image of $\man{M}$ can be
  considered as boundary points. Boundary points
  make up the boundary of $\man{M}$. Ideally such points should be classifiable
  as a singularity, points at infinity, regular points, and so on.
  A boundary construction is a generalisation of Penrose's conformal
  compactification.
  The modern boundaries are subtle and nuanced constructions. They
  can be categorised by the type of structure on $\man{M}$ that they use to
  perform the construction of the boundary.
  The $c$-boundary \cite{flores2011final} uses past and future sets.
  The isocausal boundary \cite{garcia2003causal} uses causality respecting
  embeddings. The null distance \cite{sormani2016null} selects a time function
  which induces a distance and then builds the Cauchy completion.
  The boundary given in \cite{Whale2014Chart} allows for any subsets of the 
  atlas satisfying a particular compatibility condition.
  The Abstract Boundary \cite{ScottSzekeres1994} uses the set of all
  open embeddings of the manifold. A good review of the field of boundary constructions,
  circa 2002, can be found in \cite{Ashley2002a}.

  The Abstract Boundary
  takes its inspiration from 
  the practice of studying global behaviour via charts.
  Examples of chart based global analysis
  can be found in every textbook on general relativity; our favourites are
  \cite{HawkingEllis1973, citeulike:553445, plebanski2006introduction}.
  The paper \cite{beyer2014exact} gives an example where a $3+1$ decomposition
  is used, and the two papers \cite{ScottSzekeres1986a, ScottSzekeres1986b} provide
  an example where the $3+1$ decomposition is not used.
  In all three papers \cite{beyer2014exact, ScottSzekeres1986a, ScottSzekeres1986b}
  different charts are constructed based on an analysis of curvature in order
  to ensure that the boundary points of the new chart somehow provide a
  better representation of the limiting behaviour of tensors related to the
  curvature. 

  The construction of the maximal
  analytic extension of the Kerr solution through the ring singularity
  \cite[Section 5.6]{HawkingEllis1973} provides another example.
  As presented in Boyer and Lindquist coordinates
  $(r, \theta, \phi, t)$, the set of points
  $\{(0, \theta, \phi, t):\theta\in(0,\pi),\,\phi\in(0,2\pi),\,t\in\R\}$
  all represent the same chart boundary point; the origin of $\R^4$.
  Analysis of the behaviour of $R_{abcd}R^{abcd}$ shows that
  this scalar has path dependent limits to the set
  $\{(0, \theta, \phi, t):\theta\in(0,\pi),\,\phi\in(0,2\pi),\,t\in\R\}$.
  Changing to Kerr-Schild coordinates resolves the path
  dependent limits and reveals a ring singularity as well as points
  towards which the metric has regular limits. The ring singularity
  and the regular points are still chart boundary points. They are not in
  the image of the chart; they are on the topological boundary of the
  image of the chart in Euclidean space. One can analytically extend through the
  regular points by constructing a new chart. The chart boundary points
  defining the ring singularity remain chart boundary points. 

  In \cite{Whale2014Chart} it is shown that any nice enough
  chart with boundary points
  induces an open embedding and vice versa. In this sense the
  Abstract Boundary uses the traditional chart based method of
  analysing global structure to build a non-chart dependent
  way of interpreting the word ``approach'' in the phrase
  ``approaches a singularity''.

  Theorem \ref{Thm:Endpoint Theorem} implies that the topological space
  generated by the Abstract Boundary \cite{Whale2014Chart, barry2011attached}
  is actually a compactification of $\man{M}$.
  Without this result there could exist sequences without accumulation points,
  or equivalently incomplete, inextendible curves, which approach no boundary point.
  If a boundary construction does not have enough boundary points then
  it might be possible that ``approach to a singularity'' could have no meaning
  for that boundary construction. This would demonstrate that more data than 
  assumed in the
  singularity theorems is needed to understand their physical implications.
  Thus Theorem \ref{Thm:Endpoint Theorem} is a foundational part of
  the Abstract Boundary and absolutely necessary for its relevance within the
  context of the singularity theorems.

  In the Kerr example the calculations using
  $R_{abcd}R^{abcd}$ are geometric but the claim that the singularity is
  a ring should really be worded, ``with respect to Kerr-Schild coordinates
  the chart boundary points at which $R_{acbd}R^{abcd}$ diverges
  is represented as a ring''. There is an example \cite{beyer2014exact}
  of a manifold in which there are two tensors such that there is
  no chart which resolves the path dependent limits of each tensor at the same
  time. This shows that the ring singularity of the Kerr solution
  is specific to $R_{abcd}R^{abcd}$ and the chosen chart. The ring singularity
  is not a geometric quantity; it is a representation of the values of the 
  path dependent limits of a geometric quantity in a particular chart.
  Because of the manner in which the Abstract Boundary constructs
  the larger topological space, every chart boundary point corresponds to a point in the topological space.
  This provides the statement ``The Kerr black hole contains a ring
  singularity'' with a non-chart dependent interpretation.
  It also, for example, can be used to give geometric interpretation
  to important chart dependent calculations like those in 
  \cite{klainerman2010breakdown}.

  The Abstract Boundary is the only boundary construction with a geometric
  definition of a singularity that does not rely on the introduction of
  curvature quantities \cite[Theorems 1.2 and 1.3]{whale2015generalizations}.
  Thus it aligns with the context of the Penrose-Hawking singularity theorems.
  The results contained in \cite{whale2015generalizations} rely on Theorem 
  \ref{Thm:Endpoint Theorem} to prove that every incomplete, inextendible
  curve is represented by a boundary point that can be interpreted
  as an essential singularity.
  Theorems 1.2 and 1.3 of \cite{whale2015generalizations}
  show that incomplete, inextendible curves, within very general
  classes of curves, in very general spacetimes, must end at essential Abstract
  Boundary singularities; points through which the metric cannot be extended.
  It is known that Abstract Boundary singularities satisfy generic stability
  properties \cite{Ashley2002b, FamaScott1994, FamaClarke1998}. Thus they
  are geometrically defined, physically reasonable, and closely tied to
  the incomplete, inextendible curves predicted by the singularity theorems.
  Theorems 1.2 and 1.3 of \cite{whale2015generalizations}
  therefore provide an answer to the implicit ``location'' 
  aspect of the word ``approach'' for exactly the type of singularity predicted
  by the singularity theorems.
  These theorems are vital steps towards solving the other main
  problem in the completion of the singularity theorems: “What are their
  physical properties?” The Endpoint Theorem has, and will continue to
  facilitate further research to probe the nature of these singularities,
  hopefully one day leading to the full completion of the singularity theorems;
  that is, a statement regarding the physical behaviour of the predicted
  incomplete, inextendible curves, without assumptions beyond those used in the
  singularity theorems themselves. See \cite{AshleyScott2003} for further details
  about the intended program.

\section{The Endpoint Theorem}\label{sec_theorem}
  
  Because the theorem below, in effect, attaches endpoints to
  sequences and curves we have dubbed it ``The Endpoint Theorem''.

  \begin{Thm}[The Endpoint Theorem]\label{Thm:Endpoint Theorem}
    Let $\man{M}$ and $\man{M}_\phi$ be smooth,
    connected, Hausdorff, paracompact manifolds of
    dimension $n$. If $(x_i)_{i\in\mathbb{N}}$ is a sequence of points in
    $\man{M}$ without an accumulation point, then there exists an open
    embedding $\phi:\man{M}\to\man{M}_\phi$, such that
    $\partial\phi(\man{M})$ is diffeomorphic
    to the $n-1$ dimensional unit ball  and the
    sequence $(\phi(x_i))_{i\in\mathbb{N}}$ converges to some
    $y\in\bound{\phi(\man{M})}$.
  \end{Thm}

  The proof of this result proceeds in three steps.  The first step is to
  construct a non-self-intersecting curve, $\lambda$ in $\man{M}$ so that
  $(x_i)\subset\lambda$.  The second step is to construct a
  coordinate system about this curve. In the third step we use this
  coordinate system to construct an open embedding with the required properties.

  \begin{proof}
    Let $h$ be a complete Riemannian metric on $\man{M}$. In order to
    construct a non-self-intersecting curve $\lambda$ so that
    $(x_i)\subset\lambda$ we need to construct a cover of
    $\man{M}$ by a collection $V_i$ of connected open
    submanifolds of $\man{M}$ with compact closure so that
    for all $i>0$,
    \[
      \overline{V}_{i-1}\subset V_{i}\quad\text{ and }\quad
      x_i\in V_i\backslash\overline{V}_{i-1}.
    \]
    We shall do this inductively. 
    
    Before we give the induction we define a sequence $(r_i)_{i\in\mathbb{N}}$
    in $\R^+$.
    Given $p\in\mathcal{M}$ and $r\in\mathbb{R}^+$ let
    $B_{h,r}(p)$ be the open ball of radius $r$ centred on $p$
    with respect to the complete distance, $d_h$, induced by $h$.
    For each $i\in\N$ define
    \[
      r_i=\frac{1}{2}\min_{k>i}\left\{d_h(x_k, x_0)\right\}>0.
    \]
    Since $h$ is complete and $(x_i)$ has no accumulation points,
    $\lim_{i\to\infty}r_i=\infty$. In particular,
    for each $i\in\N$, $k>i$ implies that $x_k\not\in B_{h,r_i}(x_0)$ and
    $\bigcup_iB_{h,r_i}(x_0)=\man{M}$.

    For the base case of the induction we define $V_0,V_1\subset \man{M}$ so
    that $\overline{V}_0\subset V_1$ and 
    $x_1\in V_1\setminus\overline{V}_0$. Figure \ref{figure_1}
    provides an illustration of the construction.
    Let $V_0= B_{h,r_0}(x_0)$.
    Note that
    for all $k>0$, $x_k\not\in \overline{V}_0$.
    Let $\epsilon_1<\frac{1}{2}\min_{k>1}\{d_h(x_k,\overline{V}_{0})\}$
    and define 
    $V_1^{\epsilon_1}=\{x\in\man{M}:d_h(x,\overline{V}_{0})<\epsilon_1\}$.
    Note that since $\overline{V}_0$ is connected and compact, and
    $h$ is complete, $V_1^{\epsilon_1}\setminus\overline{V}_0$ is path
    connected.
    By construction, for all $k>1$, 
    $x_k\not\in V_1^{\epsilon_1}\cup B_{h,r_1}(x_0)$.
    Choose a curve $c_1:[0,1]\to\man{M}$ so that 
    $c_1(0)\in V_1^{\epsilon_1}\setminus\overline{V}_0$, $c_1(0)\neq x_1$,
    $c_1(1)=x_1$ and
    $c_1((0,1))\cap\left(\{x_k:k>1\}\cup\overline{V}_0\right)=\EmptySet$.
    Let
    \[
      \zeta_1=\frac{1}{2}\min\left\{\min_{k>1}\{d_h(x_k,c_1)\},
        d(c_1, \overline{V}_0)\right\}
    \]
    and
    \[
      N(c_1,\zeta_1)=\{x\in\man{M}:d_h(x,c_1)<\zeta_1\}.
    \]
    Note that $N(c_1, \zeta_1)\cap\overline{V}_0=\EmptySet$.
    Define $V_1=V_1^{\epsilon_1}\cup B_{h,r_1}(x_0)\cup
    N(c_1,\zeta_1)$.

    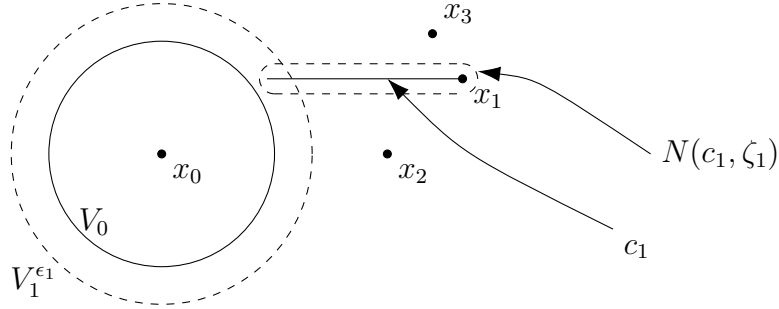
\begin{figure}
      \centering
      \begin{tikzpicture}
        \filldraw (0, 0) node[below right](0) {$x_0$} circle (1.5pt);
        \filldraw (3, 0) node[below right](2) {$x_2$} circle (1.5pt);
        \filldraw (4, 1) node[below right](1) {$x_1$} circle (1.5pt);
        \filldraw (3.6, 1.6) node[above right](3) {$x_3$} circle (1.5pt);

        \node at (-0.9,-0.9) {$V_0$};
        \node at (-1.7,-1.7) {$V_1^{\epsilon_1}$};

        \draw[-{Latex[length=3mm,width=2mm]}] (6,-1) node[below right] {$c_1$} .. controls (4,0) .. (3,1);
        \draw[-{Latex[length=3mm,width=2mm]}] (6.5,0) node[right] {$N(c_1,\zeta_1)$} .. controls (5,1) .. (4.2,1.1);

        \draw (0,0) circle (1.5cm);
        \draw[dashed] (0,0) circle (2cm);
        \draw (1.4,1) -- (4,1);
        \draw[dashed] (1.5,1.2) -- (4,1.2);
        \draw[dashed] (1.5,0.8) -- (4,0.8);
        \draw[dashed] (1.5,0.8) arc(270:90:0.2);
        \draw[dashed] (4,1.2) arc(270:90:-0.2);
      \end{tikzpicture}
      \caption{An illustrative diagram of the construction of $V_1$ given
      $V_0$. The solid circle is $V_0$. The dashed circle is $V_1^{\epsilon_1}$.
      The straight line is $c_1$ and the dashed rectangle, with curved ends,
      is $N(c_1,\zeta_1)$. We have not included $B_{h,r_1}(x_0)$ since,
      for this particular example,
      in the construction of $V_1$, 
      $B_{h,r_1}(x_0)=V_0$.
      Here $V_1=V_1^{\epsilon_1}\cup N(c_1,\zeta_1)$.}
      \label{figure_1}
    \end{figure}

    Since $V_1$ is the union of open subsets of $\man{M}$ it is an open
    submanifold. By construction, each subset $V_1^{\epsilon_1}, B_{h,r_1}(x_0)$
    and $N(c_1,\zeta_1)$ is connected. Since 
    $x_0\in V_1^{\epsilon_1}\cap B_{h,r_1}(x_0)$, and
    $c_1(0)\in N(c_1,\zeta_1)\cap V_1^{\epsilon_1}$,
    $V_1$ itself is connected.
    Since 
    each of $V_1^{\epsilon_1}, B_{h,r_1}(x_0)$
    and $N(c_1,\zeta_1)$ has compact closure, $V_1$ has compact closure.
    By definition $\overline{V}_0\subset V_1^{\epsilon_1}\subset V_1$.
    By construction $x_1=c_1(1)\in N(c_1,\zeta_1)\subset V_1$ and by construction
    $x_1\not\in \overline{V}_0$. Hence
    $x_1\in V_1\setminus\overline{V}_0$. 
    Let $k>1$. By construction $d_h\left(x_k,\overline{V}_0\right)>\epsilon_1$
    so that $x_k\not\in \overline{V}_{1}^{\,\epsilon_1}$. Similarly,
    $d_h\left(x_k,c_1\right)>\zeta_1$ so that $x_k\not\in \overline{N}(c_1,\zeta_1)$.
    Since $d_h\left(x_k,x_0\right)>r_1$ it is also the case that
    $x_k\not\in \overline{B}_{h,r_1}(x_0)$. Thus $x_k\not\in \overline{V}_1$.

    We now handle the inductive case. 
    Let $j\in\N$, $j\neq 0$.
    The argument is the same as for the base case.
    Suppose that $V_0,\ldots, V_{i-1}$ exist and are such that
    for all $j=1,\ldots,i-1$ 
    $\overline{V}_{j-1}\subset V_j$, $x_j\in V_j\setminus\overline{V}_{j-1}$
    and for all $k>j$, $x_k\not\in \overline{V}_j$.
    Let $\epsilon_i<\frac{1}{2}\min_{k>i}\{d_h(x_k,\overline{V}_{i-1})\}$.
    Define
    $V_i^{\epsilon_i}=\{x\in\man{M}:d_h(x,\overline{V}_{i-1})<\epsilon_i\}$.  
    Note that since $\overline{V}_{i-1}$ is connected and compact, and
    $h$ is complete, $V_i^{\epsilon_i}\setminus\overline{V}_{i-1}$ is path
    connected.
    By
    construction, for all $k>i$, $x_k\not\in V_i^{\epsilon_i}\cup
    B_{h,r_i}(x_0)$.
    Choose a curve $c_i:[0,1]\to\man{M}$ so that $c_i(0)\in
    V_i^{\epsilon_i}\setminus\overline{V}_{i-1}$, $c_i(0)\neq x_i$, $c_i(1)=x_i$ and
    $c_i((0,1))\cap\left(\{x_k:k>i\}\cup\overline{V}_{i-1}\right)=\EmptySet$. Let
    \[
      \zeta_i=\frac{1}{2}\min\left\{\min_{k>i}\{d_h(x_k,c_i)\},
        d(c_i,\overline{V}_{i-1})\right\}
    \]
    and
    \[
      N(c_i,\zeta_i)=\{x\in\man{M}:d_h(x,c_i)<\zeta_i\}.
    \]
    Note that $N(c_i,\zeta_i)\cap\overline{V}_{i-1}=\EmptySet$.
    Define $V_i=V_i^{\epsilon_i}\,\cup\,B_{h,r_i}(x_0)\,\cup\,
    N(c_i,\zeta_i)$.
    To complete the proof we need to prove that $V_i$ is a connected open
    submanifold of $\man{M}$ 
    with compact closure so that $\overline{V}_{i-1}\subset V_i$,
    $x_i\in V_i\setminus \overline{V}_{i-1}$ 
    and that for all $k>i$, $x_k\not\in \overline{V}_i$.

    Since $V_i$ is the union of open subsets of $\man{M}$ it is an open
    submanifold. By construction,
    each subset $V_i^{\epsilon_i}, B_{h,r_i}(x_0)$
    and $N(c_i,\zeta_i)$ is connected. Since 
    $x_0\in V_i^{\epsilon_i}\cap B_{h,r_i}(x_0)$, and 
    $c_i(0)\in N(c_i,\zeta_i)\cap V_i^{\epsilon_i}$,
    $V_i$ itself is connected.
    Since 
    each of $V_i^{\epsilon_i}, B_{h,r_i}(x_0)$
    and $N(c_i,\zeta_i)$ has compact closure, $V_i$ has compact closure.
    By definition $\overline{V}_{i-1}\subset V_i^{\epsilon_i}\subset V_i$.
    By construction $x_i=c_i(1)\in N(c_i,\zeta_i)\subset V_i$ and by assumption
    for all $j=0,\ldots,i-1$, $x_i\not\in \overline{V}_j$. Hence
    $x_i\in V_i\setminus\overline{V}_{i-1}$. 
    Let $k>i$. By construction $d_h\left(x_k,\overline{V}_{i-1}\right)>\epsilon_i$
    so that $x_k\not\in \overline{V}_{i}^{\,\epsilon_i}$. Similarly,
    $d_h\left(x_k,c_i\right)>\zeta_i$ so that $x_k\not\in \overline{N}(c_i,\zeta_i)$.
    Since $d_h\left(x_k,x_0\right)>r_i$ it is also the case that
    $x_k\not\in \overline{B}_{h,r_i}(x_0)$. Thus $x_k\not\in \overline{V}_i$.

    From above, $\left\{B_{h,r_i}(x_0):i\in\mathbb{N}\right\}$ is a cover of
    $\man{M}$. Since, for each $i\in\mathbb{N}$, 
    $B_{h,r_i}(x_0)\subset V_i$ we know that
    $\left\{V_i:i\in\mathbb{N}\right\}$ is also a cover of $\man{M}$.
    
    Now, let $\lambda_0$ be a non-self-intersecting 
    curve in $V_1$ which joins $x_0$ to
    $x_1$. For each $i>0$, $x_i\in V_{i}\setminus \overline{V}_{i-1}$.
    Thus, for each $i>0$, $x_i, x_{i+1}\in V_{i+1}\setminus\overline{V}_{i-1}$.
    We seek to show that each $V_{i+1}\setminus \overline{V}_{i-1}$ is connected
    as this will allow us to find a suitable curve from $x_i$ to $x_{i+1}$.

    We begin by showing that $V_i\setminus\overline{V}_{i-1}$ is
    connected, where 
    $V_i\setminus \overline{V}_{i-1}
      = V_i^{\epsilon_i}\setminus\overline{V}_{i-1}\cup 
        B_{h,r_i}(x_0)\setminus\overline{V}_{i-1}\cup
        N(c_i,\zeta_i)\setminus\overline{V}_{i-1}$.
    As mentioned above,
    $V_{i}^{\epsilon_i}\setminus\overline{V}_{i-1}$ is connected.
    By construction, for each $i>0$, 
    $N(c_i, \zeta_i)\cap\overline{V}_{i-1}=\EmptySet$
    and 
    $c_i(0)\in V_i^{\epsilon_i}\setminus \overline{V}_{i-1}\cap N(c_i,\zeta_i)$.
    Thus,
    $\left(V_i^{\epsilon_i}\cup N(c_i,\zeta_i)\right)\setminus\overline{V}_{i-1}$
    is connected.
    If $B_{h,r_i}(x_0)\subset V_{i-1}$, then 
    $V_i\setminus\overline{V}_{i-1}=
      \left(
        V_i^{\epsilon_i}\cup N(c_i, \zeta_i)
      \right)\setminus\overline{V}_{i-1}$
    which is connected.
    Otherwise, there exists $y_1\in B_{h,r_i}(x_0)\setminus \overline{V}_{i-1}$.
    Let $\gamma_1$ be a non-self-intersecting curve in $B_{h, r_i}(x_0)$
    which joins $x_0$ to $y_1$.
    There must exist a point 
    $p_1\in\gamma_1\setminus\overline{V}_{i-1} \cap 
      V_i^{\epsilon_i}\setminus\overline{V}_{i-1}$.
    Since $y_1$ is path connected to $p_1$ in 
    $B_{h,r_i}(x_0)\setminus\overline{V}_{i-1}$, $y_1$ is path connected
    in $V_i\setminus\overline{V}_{i-1}$ to any point in
    $\left(V_i^{\epsilon_i}\cup N(c_i,\zeta_i)\right)\setminus\overline{V}_{i-1}$.
    Let $y_2\in B_{h,r_i}(x_0)\setminus\overline{V}_{i-1}$ where
    $y_1$ and $y_2$ are distinct points. Let
    $\gamma_2$ be a non-self-intersecting curve in $B_{h,r_i}(x_0)$
    which joins $x_0$ to $y_2$, where $\gamma_1\cap\gamma_2=\EmptySet$.
    There must exist a point 
    $p_2\in\gamma_2\setminus\overline{V}_{i-1}\cap
      V_i^{\epsilon_i}\setminus\overline{V}_{i-1}$. Since
    $y_2$ is path connected to $p_2$ in $B_{h,r_i}(x_0)\setminus\overline{V}_{i-1}$,
    $y_2$ is path connected in $V_i\setminus\overline{V}_{i-1}$ to
    $p_1$ and $y_1$. Thus $V_i\setminus\overline{V}_{i-1}$ is connected.

    We now show that $V_{i+1}\setminus \overline{V}_{i-1}$ is connected.
    The paragraph above implies that $V_{i+1}\setminus\overline{V}_i$ and
    $V_i\setminus\overline{V}_{i-1}$ are connected. Consider a point
    $p\in V_{i+1}^{\epsilon_{i+1}}\setminus\overline{V}_i$. Since
    $V_{i+1}^{\epsilon_{i+1}}
      =\left\{x \in \man{M}:d_h(x,\overline{V}_i)<\epsilon_{i+1}\right\}$
    and $V_{i+1}^{\epsilon_{i+1}}\setminus\overline{V}_i$ is connected, there
    exists a path in $V_{i+1}^{\epsilon_{i+1}}\setminus\overline{V}_i$
    which connects $p$ to a point 
    $q\in\partial V_i\setminus\overline{V}_{i-1}\cap
      \partial V_{i+1}\setminus\overline{V}_i$.
    Due to the connectedness of $V_i\setminus\overline{V}_{i-1}$,
    $q$ and thus $p$ are path connected in $V_{i+1}\setminus\overline{V}_{i-1}$
    to every point in $V_i\setminus\overline{V}_{i-1}$ and
    $\partial V_i\setminus\overline{V}_{i-1}\cap
      \partial V_{i+1}\setminus\overline{V}_i$.
    This implies that $V_{i+1}\setminus \overline{V}_{i-1}$ is connected.

    We now introduce the required curve from $x_{i}$ to $x_{i+1}$.
    Since $V_{i+1}\setminus\overline{V}_{i-1}$ is connected
    there exists a non-self-intersecting curve, $\lambda_i$, in
    $V_{i+1}\setminus\overline{V}_{i-1}$ that joins $x_i$ and $x_{i+1}$.
    We may do this in such a way that
    $\lambda_i\cap\lambda_{i-1}=x_i$. By
    `smoothing' at the joins we get a $C^\infty$ curve
    $\lambda:[0,1)\to\man{M}$ which is non-self-intersecting, such that
    $(x_i)\subset \lambda$. Note that $\lambda$ may need to be
    rescaled so that $\lambda'(t)\neq0$ for $t\in[0,1)$.
    Let $(t_i)_{i\in\mathbb{N}}$ be a sequence in
    $[0,1)$ so that $t_0=0$, and for all $i,j\in\mathbb{N}$, $\lambda(t_i)=x_i$,
    $t_i<t_j \Leftrightarrow i<j$, and $t_i\to 1$ as $i\to\infty$.
    By construction, for each $i\in\N$, 
    $\lambda$ eventually leaves, and never returns to,
    $\overline{V_i}$. Since $\{\overline{V_i}:i\in\N\}$ is a compact
    exhaustion of $\man{M}$ we know that $d_h(x_0,\lambda(t))\to\infty$ as
    $t\to 1$. This implies that the non-self-intersecting curve $\lambda$
    has no accumulation points and therefore approaches the ``edge'' of the
    manifold as $t\to 1$.

    Let $N\lambda=\{v\in T_{\lambda(t))}\man{M}:h(v,\lambda'(t))=0, t\in [0,1)\}$
    be the normal bundle of $\lambda$. It is a sub-bundle of
    $T\man{M}$, hence we can restrict $h$ to $N\lambda$. The restriction of
    the metric $h$ induces a Levi-Civita connection
    on $N\lambda$.
    Let
    $\{E_i\}_{i=1,\ldots,n-1}$ be a linearly independent
    frame at $N_{\lambda(0)}\lambda$.
    Parallelly propagate, with respect to the induced Levi-Civita connection,
    this frame along $\lambda$ in $N\lambda$ to get $n-1$
    linearly independent
    vector fields $E_i:[0,1)\to N\lambda$.
    The bundle $N\lambda$ is an inclusion into $TM$.
    In an abuse of notation denote the image of $E_i$ under this
    inclusion by $E_i$. By construction
    for all $i=1,\ldots, n-1$, $h(E_i, \lambda')=0$.

    Let $B_{n-1}$ be the unit ball in $\R^{n-1}$. The existence of
    a normal neighbourhood of $\lambda$, \cite[Proposition 7.26]{ONeil1983},
    implies that there exists a smooth function $f:(0,1)\to\mathbb{R}^+$
    so that the function
    $\mu:(0,1)\times B_{n-1}\to\man{M}$ defined by
    $\mu(t,p)=\exp_{\lambda(t)}(f(t)p^i E_i(t))$, where
    $i$ sums over $1$ to $n-1$, is a chart.

    Now we use the chart $\mu$ to construct an open embedding of \man{M}. Let
    $\man{N}$ be the set $\man M\cup\left([1,2)\times B_{n-1}\right)$. Give
    $\man{N}$ the atlas generated by the set containing all the charts of $\man{M}$ 
    as well
    as a new chart $\psi:(0,2)\times B_{n-1}\to \man{N}$ given by
    \begin{gather}
      \psi(t,p)=\left\{
      \begin{aligned}
        &\exp_{\lambda(t)}(f(t)p^i E_i(t)) \quad\quad\text{if}\ t\in(0,1)\\
        &(t,p) \quad\quad\quad\quad\quad\quad\quad\quad\ \text{otherwise.}
      \end{aligned}\right.
    \end{gather}

    The only non-trivial part is to show that $\man N$ is a Hausdorff manifold. 
    This reduces to showing that the
    points $\{1\}\times B_{n-1}$ are Hausdorff separated from
    \man{M}. This must be true, however, since for any
    $p\in\man{M}$ there exists $i$ so that $p\in V_{i}$, 
    and $\man{N}\setminus\overline{V_i}$ is an open neighbourhood containing every
    point in $\{1\}\times B_{n-1}$. Thus $\man{N}$ is a smooth, connected,
    Hausdorff, paracompact manifold of dimension $n$.
    By definition $\partial\man{M}$ considered
    as a subset of $\man{N}$ is $\{1\}\times B_{n-1}$ which is diffeomorphic
    to the $n-1$ dimensional unit ball.

    Let $\phi:\man{M}\to\man{N}$ be the identity. Then $\phi$ is
    an open embedding of $\man{M}$ and the sequence $(\phi(x_i))_{i\in\N}$
    converges to the point $(1,0)\in [1,2)\times B_{n-1}$ of
    $\man{N}$, as required.
  \end{proof}

  \begin{Cor}
    Let $\man{M}$ and $\man{M}_\phi$
    be smooth, connected, Hausdorff, paracompact manifolds of dimension $n$. 
    Let
    $\gamma:[a,b)\to\man{M}$ be a non-self-intersecting curve in
    $\man{M}$
    without limit points in
    $\man{M}$. Then there exists an open embedding
    $\phi:\man{M}\to\man{M}_\phi$ so that the curve $\phi(\gamma)$
    has an endpoint in $\bound{\phi(\man{M})}$.
  \end{Cor}
  \begin{proof}
    Since $\gamma$ is non-self-intersecting, we may use it as the
    curve in the proof of the Endpoint Theorem.  The construction
    of $\phi$ then implies that $\phi(\gamma)\to(1,0)\in\man{N}$.
  \end{proof}

\ack
  The authors give their thanks to Christopher J.\ S.\ Clarke who worked with 
  Susan M. Scott
  on the original development of this result and sadly passed away on
  16 April 2019. As a student, Scott was inspired by the work of Clarke on singularities,
  and they subsequently collaborated in this field a number of times, including for the
  production of the Endpoint Theorem presented here.
  
\section{Bibliography}

 \bibliographystyle{unsrt}
 \bibliography{Thesis-main}

\begin{thebibliography}{10}

\bibitem{Clarke1993}
C.~J.~S. Clarke.
\newblock {\em The analysis of space-time singularities}, volume~1 of {\em
  Cambridge Lecture Notes in Physics}.
\newblock Cambridge University Press, Cambridge, 1993.

\bibitem{SlupinskiClarke1980}
M.~J. Slupinski and C.~J.~S. Clarke.
\newblock Singular points and projective limits in relativity.
\newblock {\em Communications in Mathematical Physics}, 71(3):289--297, 1980.

\bibitem{TiplerClarkeEllis1980}
F.~J. Tipler, C.~J.~S. Clarke, and G.~F.~R. Ellis.
\newblock Singularities and horizons - a review article.
\newblock In A.~{Held} and J.~L. {Anderson}, editors, {\em {General Relativity
  and Gravitation: One Hundred Years After the Birth of Albert Einstein}},
  volume~2. Plenum Press, New York, 1980.

\bibitem{Clarke1979}
C.~J.~S. Clarke.
\newblock Boundary definitions.
\newblock {\em General Relativity and Gravitation}, 10(12):977--980, 1979.

\bibitem{1978CMaPh..58..291C}
C.~J.~S. {Clarke}.
\newblock {The singular holonomy group}.
\newblock {\em Communications in Mathematical Physics}, 58:291--297, 1978.

\bibitem{1975GReGr...6...35C}
C.~J.~S. {Clarke}.
\newblock {The classification of singularities}.
\newblock {\em General Relativity and Gravitation}, 6:35--40, 1975.

\bibitem{MR0371347}
C.~J.~S. Clarke.
\newblock Singularities in globally hyperbolic space-time.
\newblock {\em Communications in Mathematical Physics}, 41:65--78, 1975.

\bibitem{ClarkeSchmidt1977}
C.~J.~S. Clarke and B.~G. Schmidt.
\newblock Singularities - {S}tate of art.
\newblock {\em General Relativity and Gravitation}, 8(2):129--137, 1977.

\bibitem{Clarke1983}
C.~J.~S. Clarke.
\newblock The cardinality of manifold atlases.
\newblock {\em Israel Journal of Mathematics}, 45(1):9--16, 1983.

\bibitem{FamaClarke1998}
C.~J. Fama and C.~J.~S. Clarke.
\newblock A rigidity result on the ideal boundary structure of smooth
  spacetimes.
\newblock {\em Classical and Quantum Gravity}, 15(9):2829--2840, 1998.

\bibitem{ClarkeKrolak1985}
C.~J.~S. Clarke and A.~Kr{\'o}lak.
\newblock Conditions for the occurrence of strong curvature singularities.
\newblock {\em Journal of Geometry and Physics}, 2(2):127--143, 1985.

\bibitem{Penrose1972}
R.~Penrose.
\newblock {\em Techniques of Differential Topology in Relativity}.
\newblock Society for Industrial and Applied Mathematics, 1987.
\newblock Conference Board of the Mathematical Sciences Regional Conference
  Series in Applied Mathematics, No. 7.

\bibitem{BeemEhrlichEasley1996}
J.~K. Beem, P.~E. Ehrlich, and K.~L. Easley.
\newblock {\em Global {L}orentzian {G}eometry}, volume 202 of {\em Pure and
  Applied Mathematics: A Series of Monographs and Textbooks}.
\newblock Marcel Dekker, Inc., 1996.

\bibitem{senovilla1998singularity}
J.~M.~M. Senovilla.
\newblock Singularity theorems and their consequences.
\newblock {\em General Relativity and Gravitation}, 30(5):701--848, 1998.

\bibitem{senovilla20151965}
J.~M.~M. Senovilla and D.~Garfinkle.
\newblock The 1965 {P}enrose singularity theorem.
\newblock {\em Classical and Quantum Gravity}, 32(12):124008, 2015.

\bibitem{HawkingEllis1973}
S.~W. Hawking and G.~F.~R. Ellis.
\newblock {\em The Large Scale Structure of Space-Time}.
\newblock Cambridge University Press, 1973.

\bibitem{penrose1965gravitational}
R.~Penrose.
\newblock Gravitational collapse and space-time singularities.
\newblock {\em Physical Review Letters}, 14(3):57--59, 1965.

\bibitem{lifshitz1963investigations}
E.~M. Lifshitz and I.~M. Khalatnikov.
\newblock Investigations in relativistic cosmology.
\newblock {\em Advances in Physics}, 12(46):185--249, 1963.

\bibitem{0038-5670-13-6-R04}
V.~A. Belinskii, E.~M. Lifshitz, and I.~M. Khalatnikov.
\newblock Oscillatory approach to a singular point in relativistic cosmology.
\newblock {\em Soviet Physics Uspekhi}, 13(6):745--765, 1971.

\bibitem{sormani2016null}
C.~Sormani and C.~Vega.
\newblock Null distance on a spacetime.
\newblock {\em Classical and Quantum Gravity}, 33(8):085001, 2016.

\bibitem{flores2011final}
J.~L. Flores, J.~Herrera, and M.~S{\'a}nchez.
\newblock On the final definition of the causal boundary and its relation with
  the conformal boundary.
\newblock {\em Advances in Theoretical and Mathematical Physics},
  15(4):991--1057, 2011.

\bibitem{garcia2003causal}
A.~Garcia-Parrado and J.~M.~M. Senovilla.
\newblock Causal relationship: a new tool for the causal characterization of
  {Lorentzian} manifolds.
\newblock {\em Classical and Quantum Gravity}, 20(4):625, 2003.

\bibitem{ScottSzekeres1994}
S.~M. Scott and P.~Szekeres.
\newblock The {A}bstract {B}oundary---a new approach to singularities of
  manifolds.
\newblock {\em Journal of Geometry and Physics}, 13(3):223--253, 1994.

\bibitem{Whale2014Chart}
B.~E. Whale.
\newblock The chart based approach to studying the global structure of a
  spacetime induces a coordinate invariant boundary.
\newblock {\em General Relativity and Gravitation}, 46(1):1--43, 2014.

\bibitem{Ashley2002a}
M.~J. S.~L. Ashley.
\newblock {\em Singularity Theorems and the Abstract Boundary Construction}.
\newblock PhD thesis, Department of Physics, Australian National University,
  2002.
\newblock Located at \textsf{http://hdl.handle.net/1885/46055}.

\bibitem{citeulike:553445}
C.~W. Misner, K.~S. Thorne, and J.~A. Wheeler.
\newblock {\em Gravitation (Physics Series)}.
\newblock {W. H. Freeman}, 1973.

\bibitem{plebanski2006introduction}
J.~Plebanski and A.~Krasinski.
\newblock {\em An introduction to general relativity and cosmology}.
\newblock Cambridge University Press, 2006.

\bibitem{beyer2014exact}
F.~Beyer and J.~Hennig.
\newblock An exact smooth gowdy-symmetric generalized taub--nut solution.
\newblock {\em Classical and Quantum Gravity}, 31(9):095010, 2014.

\bibitem{ScottSzekeres1986a}
S.~M. Scott and P.~Szekeres.
\newblock The {C}urzon singularity. {I}: Spatial sections.
\newblock {\em General Relativity and Gravitation}, 18:557--570, 1986.

\bibitem{ScottSzekeres1986b}
S.~M. Scott and P.~Szekeres.
\newblock The {C}urzon singularity. {II}: Global picture.
\newblock {\em General Relativity and Gravitation}, 18:571--583, 1986.

\bibitem{barry2011attached}
R.~A. Barry and S.~M. Scott.
\newblock The attached point topology of the abstract boundary for spacetime.
\newblock {\em Classical and Quantum Gravity}, 28(16):165003, 2011.

\bibitem{klainerman2010breakdown}
S.~Klainerman and I.~Rodnianski.
\newblock On the breakdown criterion in general relativity.
\newblock {\em Journal of the American Mathematical Society}, 23(2):345--382,
  2010.

\bibitem{whale2015generalizations}
B.~E. Whale, M.~J.~S.~L. Ashley, and S.~M. Scott.
\newblock {G}eneralizations of the {A}bstract {B}oundary singularity theorem.
\newblock {\em Classical and Quantum Gravity}, 32(13):135001, 2015.

\bibitem{Ashley2002b}
M.~J. S.~L. Ashley.
\newblock The stability of abstract boundary essential singularities.
\newblock {\em General Relativity and Gravitation}, 34(10):1625--1635, 2002.

\bibitem{FamaScott1994}
C.~J. Fama and S.~M. Scott.
\newblock Invariance properties of boundary sets of open embeddings of
  manifolds and their application to the {A}bstract {B}oundary.
\newblock In {\em Differential geometry and mathematical physics (Vancouver,
  BC, 1993)}, volume 170 of {\em Contemporary Mathematics}, pages 79--111.
  American Mathematical Society, Providence, RI, 1994.

\bibitem{AshleyScott2003}
M.~J. S.~L. Ashley and S.~M. Scott.
\newblock Curvature singularities and abstract boundary singularity theorems
  for space-time.
\newblock In K.~L. Duggal and R.~Sharma, editors, {\em Recent advances in
  Riemannian and Lorentzian geometries}, volume 337 of {\em Contemporary
  Mathematics}, pages 9--19. American Mathematical Society, 2003.

\bibitem{ONeil1983}
B.~O'Neill.
\newblock {\em Semi-Riemannian Geometry with Applications to Relativity},
  volume 103 of {\em Pure and Applied Mathematics}.
\newblock Academic Press, 1983.

\end{thebibliography}

\end{document}